\documentclass[draftcls,onecolumn,11pt]{IEEEtran}
\usepackage{amsmath, amssymb}
\usepackage{graphicx, color}
\usepackage{bm}
\usepackage{cite}
\graphicspath{{figure/}}

\newtheorem{ass}{\bf Assumption}
\newtheorem{lma}{\bf Lemma}
\newtheorem{thm}{\bf Theorem}

\begin{document}

\title{Multiobjective Controller Design by Solving a Multiobjective Matrix Inequality Problem}

\author{Wei-Yu~Chiu
\thanks{W.-Y. Chiu is with the Multiobjective Control Lab, Department of Electrical Engineering, Yuan Ze University, Taoyuan 32003, Taiwan
        (email: wychiu@saturn.yzu.edu.tw).}%
\thanks{This work was supported by the Ministry of Science and Technology of Taiwan under Grant 102-2218-E-155-004-MY3.}
\thanks{This paper is a postprint of a paper submitted to and accepted for publication in IET Control Theory and Applications and is subject to Institution of Engineering and Technology Copyright. The copy of record is available at IET Digital Library.}
\thanks{doi: 10.1049/iet-cta.2014.0026}
 }
\maketitle

\begin{abstract}
In this study, linear matrix inequality (LMI) approaches and multiobjective (MO) evolutionary algorithms are integrated to design controllers.
An MO matrix inequality problem (MOMIP) is first defined.
A hybrid MO differential evolution (HMODE) algorithm is then developed to solve the MOMIP.
The hybrid algorithm combines deterministic and stochastic searching schemes.
In the solving process, the deterministic part aims to exploit the structures of matrix inequalities, and
the stochastic part is used to  fully explore the decision variable space.
 Simulation results show that the HMODE algorithm can produce an approximated Pareto front (APF)
and Pareto-efficient controllers that stabilize the associated controlled system.
In contrast with single-objective designs using LMI approaches,
the proposed MO methodology can clearly  illustrate how the objectives involved affect each other, i.e.,
a broad perspective on optimality is provided. This facilitates the selecting process for a representative design, and particularly the design that corresponds to a nondominated vector lying in the knee region of the APF. In addition, controller gains can be readily modified to incorporate the preference or need of a system designer.
\end{abstract}

\newpage

\section{Introduction}\label{sec_intro}

Controller design problems with multiple objectives  have been extensively investigated because real-world designs mostly involve multiple objectives that need to be achieved~\cite{app1,app2,appli5}.
In the literature, attaining multiple objectives in a control system has been interpreted differently, yielding various multiobjective (MO) design approaches.
For example, in~\cite{appli1} and \cite{appli6},
an MO controller design  means that
the resulting controller should satisfy
some inequality constraints related to several aspects of system performance. In this case,
 the design procedure does not involve minimising any cost functions (or objective functions) in the objective space.
Furthermore, bounds on system specifications are often prescribed rather than being assigned through an optimisation process.

In certain scenarios, an MO design leads to solving a minimisation problem with
only one cost function~\cite{GA_based,appli2,appli4,C_2_1,decay_rate,C_2_2,C3_2}.
A series of single-objective (SO) formulations is mostly used, e.g.,  linear weighting methods, weighted geometric mean approaches,
boundary intersection approaches, and $\varepsilon$-constraint methods~\cite{decom1,decom2,decom3,decom,Coello}.
Although the original design of a control system considers the optimisation of  multiple cost functions,
  an SO problem (SOP) instead of an MO problem (MOP) is solved in the end~\cite{appli3,SO_1,SO_2}.

In contrast with the cases in which none or exactly one cost function is involved, an MO design may indeed
address two or more  cost functions in the design procedure.
Existing MO design approaches often employ advanced multiobjective evolutionary algorithms (MOEAs) to  evaluate design parameters~\cite{C3_1,app3,C3_4,review}, yielding an approximated Pareto front (APF).
System designers thus have an advantage of being able to understand how multiple objectives affect each other
as compared to SO designs.
In this scenario,
 each point or vector on the APF associates with a Pareto-efficient design.
 The availability of the APF allows the designer
 to incorporate
their preference models into the selecting process for  the ``optimal'' trade-off design.
These advantages cannot be achieved through solving a SOP.

Based on above discussions,
using an MOEA to solve an MOP for parameter evaluation can be a promising design approach.
For controller designs,  linear matrix inequality (LMI) approaches have been widely applied to system designs~\cite{Boyd_LMI}, e.g., predictive control designs~\cite{predic1,predic2},
estimation/observer designs~\cite{Hinf_chiu,case1}, and nonlinear system designs~\cite{Tanaka_fuzzy}.
It becomes constructive to combine these two tools to develop an MO design approach.
In~\cite{case1}, an attempt was made to design a networked system (not a control system) with multiple objectives,
but a few restrictions on the associated MOP were imposed.
For instance, only two objectives can be involved, and the resulting MOP must reduce to an eigenvalue problem after one objective is removed.
These restrictions impede the application of the proposed methodology in a general setting.
Further research is needed.

In this paper, we aim to integrate MOEAs and LMI approaches for MO controller designs.
To this end, a multiobjective matrix inequality problem (MOMIP) is  defined.
An MOEA termed hybrid multiobjective differential evolution (HMODE) algorithm is then proposed to solve it.
Once the MOMIP has been solved, Pareto optimality can be achieved by producing Pareto-efficient controllers.
In our framework, the form of the MOMIP results from using LMI approaches in consideration of multiple objectives.
It is basically an MOP with matrix inequalities (MIs) as constraints.
In such an MOP, two types of decision variables, matrix and non-matrix decision variables, are involved.
Regarding the hybrid algorithm, it adopts both stochastic and deterministic searching schemes, i.e.,
differential evolution (DE) algorithms and interior-point methods, respectively.
During the solving process, the former and latter schemes are used to determine
 the non-matrix and the matrix decision variables, respectively.

The main contributions of this study are as follows.
We integrate LMI approaches and MOEAs for controller designs in consideration of multiple objectives,
 which has not been thoroughly investigated in the literature.
The MOMIP is defined and the novel HMODE algorithm is proposed to solve it.
By doing so, we connect the evolutionary algorithms field with the control field
by providing a framework for the development of hybrid MOEAs that are applicable to MO controller designs.

The following notation is used throughout this paper.
The sets of positive integers and real numbers are denoted by $\mathbb{Z}_+$ and $\mathbb{R}$, respectively.
For a matrix $\bm{A}$, $\bm{A}>0$ means that $\bm{A}$ is symmetric and positive-definite.
We further define $\bm{A} < 0$ as  $-\bm{A} >0$.
To simplify our notation,
the star mark ``$\star$'' is used in the following two ways:
\begin{equation*}
    \left[
       \begin{array}{cc}
         \bm{A} & \star \\
         \bm{B} & \bm{C} \\
       \end{array}
     \right]
     =
      \left[
       \begin{array}{cc}
         \bm{A} & \bm{B}^T \\
         \bm{B} & \bm{C} \\
       \end{array}
     \right] \mbox{ and } (\bm{A},\star)=\bm{A}+\bm{A}^T
\end{equation*}
for appropriate dimensions of matrices $\bm{A},\bm{B}$, and $\bm{C}$.
For the terminology of Pareto optimality, i.e.,
the Pareto dominance, the Pareto optimal set, and the Pareto front,
the reader can refer to~\cite{Coello}.

The rest of this paper is organized as follows.
In Section~\ref{sec_MOMIP}, we define the MOMIP
and derive its equivalent form.
Design examples are presented in Section~\ref{sec_example}.
The HMODE algorithm is developed in Section~\ref{sec_HMODE}.
Finally, Section~\ref{sec_con} concludes this paper.

\section{Multiobjective Matrix Inequality Problem (MOMIP)}\label{sec_MOMIP}

 In our MO design approach, the MOMIP is defined as
\begin{equation}\label{eq_MOMIP}
 \begin{split}
   \min_{\bm{\alpha},\bm{X}} \; &  \bm{f}(\bm{\alpha})  \\
  \mbox{subject to }&  \mathcal{MI}(\bm{\alpha},\bm{X}) < 0.
 \end{split}
\end{equation}
In~(\ref{eq_MOMIP}),
\begin{equation}\label{eq_dim}
    \bm{\alpha}=  \left[
                    \begin{array}{cccc}
                     \alpha_1 & \alpha_2 & \cdots & \alpha_M \\
                    \end{array}
                  \right]^T
    \in \mathbb{R}^M\mbox{ and }
   \bm{f}(\bm{\alpha})\in \mathbb{R}^N
\end{equation}
represent the vector of all non-matrix decision variables and the vector-valued objective function, respectively.
The $\bm{X}$ represents a collection of  matrix decision variables. The $\mathcal{MI}(\bm{\alpha},\bm{X})$ denotes a matrix-valued function, and the condition $\mathcal{MI}(\bm{\alpha},\bm{X}) < 0$  represents MI constraints.
Regarding the cost function $ \bm{f}(\bm{\alpha})$, it can be
  a vector of $H_\infty$ attenuation levels when an $H_\infty$ design is adopted.
   For optimal control,
 each entry of $\bm{f}(\bm{\alpha})$ can represent an upper bound on a quadratic performance function.
 For a mixed  $H_2/H_\infty$ design,
  $\bm{f}(\bm{\alpha})$ can represent  a mix of  performance indices.
Regarding
the associated MI
\begin{equation}\label{eq_MIP}
\mathcal{MI}(\bm{\alpha},\bm{X}) <0
\end{equation}
the following assumption is made.

\begin{ass}\label{ass_1}
 The MI problem~(\ref{eq_MIP}) is convex in the variable $\bm{X}$ once $\bm{\alpha}$ is given.
\end{ass}

In Assumption~\ref{ass_1},
the convexity in $\bm{X}$ originates from the fact that many SO control problems  can be designed by solving a convex LMI problem (LMIP).
It should be noted that Assumption~\ref{ass_1} does not imply that~(\ref{eq_MIP}) is convex in both $\bm{X}$ and $\bm{\alpha}$.
It will become clear that by using LMI approaches, controller design problems with multiple objectives can naturally lead to~(\ref{eq_MOMIP}).

For the purpose of algorithm development, we derive another form that is equivalent to the MOMIP~(\ref{eq_MOMIP}).
The HMODE algorithm will be developed based on the derived form.
To begin with the derivation, let us consider conventional SO designs that adopt LMI approaches, yielding
the  eigenvalue problem (EVP)
 \begin{equation}\label{eq_EVP}
 (\lambda^*(\bm{\alpha}), \bm{X}^*(\bm{\alpha}))
  =
  \begin{array}{l}
    \arg_{\lambda,\bm{X}}  \min_{\lambda,\bm{X}} \;  \lambda  \\
    \mbox{subject to }  \mathcal{MI}(\bm{\alpha},\bm{X}) < \lambda \bm{I}.
  \end{array}
\end{equation}
In~(\ref{eq_EVP}), the bound $\lambda$ on the maximum eigenvalue
of the matrix $\mathcal{MI}(\bm{\alpha},\bm{X})$ is minimised. Here,
$\bm{X}^*(\bm{\alpha})$ represents the matrix variable that achieves the minimum $\lambda^*(\bm{\alpha})$.
It is noted that, while $\bm{\alpha}$ is the non-matrix decision variable in the MOMIP~(\ref{eq_MOMIP}), it
is not a decision variable in~(\ref{eq_EVP}).
The values of $\lambda^*(\bm{\alpha})$ and $\bm{X}^*(\bm{\alpha})$ depend on the value of $\bm{\alpha}$.
 According to Assumption~\ref{ass_1},
  the EVP~(\ref{eq_EVP}) is convex in the variable $\bm{X}$ given $\bm{\alpha}$, implying that
  it can be solved by interior-point methods.

The following lemma relates the feasibility of the MOMIP~(\ref{eq_MOMIP}) to the EVP~(\ref{eq_EVP}).
\begin{lma}\label{lma_1}
If a pair~$(\tilde{\bm{\alpha}},\tilde{\bm{X}})$ is feasible in~(\ref{eq_MOMIP}), i.e., $  \mathcal{MI}(\tilde{\bm{\alpha}},\tilde{\bm{X}}) < 0$, then $\mathcal{MI}(\tilde{\bm{\alpha}},\bm{X}^*(\tilde{\bm{\alpha}})) < 0$.
\end{lma}
\begin{proof}
Let~$\lambda_{max}$ be the maximum eigenvalue of~$\mathcal{MI}(\tilde{\bm{\alpha}},\tilde{\bm{X}})$.
The condition $\mathcal{MI}(\tilde{\bm{\alpha}},\tilde{\bm{X}}) < 0$ implies that $\lambda_{max}<0$
and thus $\mathcal{MI}(\tilde{\bm{\alpha}},\tilde{\bm{X}}) \leq \lambda_{max} \bm{I} < \frac{\lambda_{max}}{2} \bm{I}<0$.
The pair $(\frac{\lambda_{max}}{2},\tilde{\bm{X}})$ is a feasible solution of the EVP~(\ref{eq_EVP}) given $\bm{\alpha}=\tilde{\bm{\alpha}}$.
By the definitions of~$\lambda^*(\cdot)$ and~$\bm{X}^*(\cdot)$ in~(\ref{eq_EVP}),
 we have $\lambda^*(\tilde{\bm{\alpha}})\leq \frac{\lambda_{max}}{2} <0$ and $\mathcal{MI}(\tilde{\bm{\alpha}},\bm{X}^*(\tilde{\bm{\alpha}})) <  \lambda^*(\tilde{\bm{\alpha}}) \bm{I}$, which completes the proof.
\end{proof}

With the help of Lemma~\ref{lma_1}, the following theorem shows that the MOMIP
\begin{equation}\label{eq_MOMIP2}
 \begin{split}
   \min_{\bm{\alpha}} \; &  \bm{f}(\bm{\alpha})  \\
  \mbox{subject to }&  \mathcal{MI}(\bm{\alpha},\bm{X}^*(\bm{\alpha})) < 0
 \end{split}
\end{equation}
is equivalent to~(\ref{eq_MOMIP}).

\begin{thm}\label{thm_equi}
An $\bm{\alpha}^*$ is a Pareto optimal solution of~(\ref{eq_MOMIP2})
if and only if there exists a matrix $\bm{X}'$ such that $(\bm{\alpha}^*,\bm{X}')$ is a Pareto optimal solution of~(\ref{eq_MOMIP}).
\end{thm}
\begin{proof} ``$\Rightarrow$''
Since $\bm{\alpha}^*$ is Pareto optimal, it is a  feasible solution of~(\ref{eq_MOMIP2}), i.e.,  $  \mathcal{MI}(\bm{\alpha}^*,\bm{X}^*(\bm{\alpha}^*)) < 0$.
By taking $\bm{X}'=\bm{X}^*(\bm{\alpha}^*)$, the pair~$(\bm{\alpha}^*,\bm{X}')$ is a feasible solution of~(\ref{eq_MOMIP}).
We now proceed by contraposition. Suppose that there exists a feasible pair~$(\bm{\alpha}',\bm{X}'')$ of~(\ref{eq_MOMIP}) such that
 $\bm{f}(\bm{\alpha}')$ dominates $\bm{f}(\bm{\alpha}^*)$, denoted by~$\bm{f}(\bm{\alpha}')  \prec \bm{f}(\bm{\alpha}^*)$.
Since $ \mathcal{MI}(\bm{\alpha}',\bm{X}'') < 0$, according to Lemma~\ref{lma_1}, we have $ \mathcal{MI}(\bm{\alpha}',\bm{X}^*(\bm{\alpha}')) < 0$.
The conditions~$\bm{f}(\bm{\alpha}')\prec \bm{f}(\bm{\alpha}^*)$ and~$ \mathcal{MI}(\bm{\alpha}',\bm{X}^*(\bm{\alpha}')) < 0$
contradict the Pareto optimality of $\bm{\alpha}^*$ in~(\ref{eq_MOMIP2}). Therefore,
such a pair~$(\bm{\alpha}',\bm{X}'')$ does not exist, and  the pair $(\bm{\alpha}^*,\bm{X}')$ must be Pareto optimal in~(\ref{eq_MOMIP}).

``$\Leftarrow$''
Since $(\bm{\alpha}^*,\bm{X}')$ is Pareto optimal, it is feasible in~(\ref{eq_MOMIP}) and, according to Lemma~\ref{lma_1}, ~$\bm{\alpha}^*$ is a feasible solution of~(\ref{eq_MOMIP2}).
We proceed by contraposition. Suppose that~$\bm{\alpha}^*$ is not Pareto optimal in~(\ref{eq_MOMIP2}).
There exists a  feasible~$\bm{\alpha}'$, i.e., $ \mathcal{MI}(\bm{\alpha}',\bm{X}^*(\bm{\alpha}') ) < 0$, such that
 $\bm{f}(\bm{\alpha}')$ dominates $\bm{f}(\bm{\alpha}^*)$, denoted by~$\bm{f}(\bm{\alpha}')  \prec \bm{f}(\bm{\alpha}^*)$.
Let $\bm{X}''=\bm{X}^*(\bm{\alpha}')$.  We have
$ \mathcal{MI}(\bm{\alpha}',\bm{X}'' ) < 0$ and~$\bm{f}(\bm{\alpha}')  \prec \bm{f}(\bm{\alpha}^*)$,
which contradicts the fact that the pair $(\bm{\alpha}^*,\bm{X}')$ is Pareto optimal in~(\ref{eq_MOMIP}).
Therefore, the feasible~$\bm{\alpha}'$ that dominates~$\bm{\alpha}^*$ in~(\ref{eq_MOMIP2})
does not exist. We conclude that $\bm{\alpha}^*$ is Pareto optimal in~(\ref{eq_MOMIP2}).
\end{proof}

Theorem~\ref{thm_equi} is  key to the successful application of MOEAs to solving MOMIPs.
As  the MOMIP~(\ref{eq_MOMIP2}) is equivalent to~(\ref{eq_MOMIP}) according to Theorem~\ref{thm_equi},
we can solve~(\ref{eq_MOMIP2}) to obtain the design parameters.
It is  worth noting that $\bm{X}$ and $\bm{\alpha}$ are the decision variables in~(\ref{eq_MOMIP}) while  $\bm{\alpha}$ is the only decision variable
in~(\ref{eq_MOMIP2}). To some extent, the MOMIP~(\ref{eq_MOMIP2}) can be regarded as a conventional MOP with $\bm{\alpha}$ as the decision variable.

\section{Design Examples}\label{sec_example}

In this section, design examples associated with the MOMIP~(\ref{eq_MOMIP2}) are investigated.

\subsection{Example 1: Robust Control Design}\label{subsec_uncertain}

Consider a robust $H_\infty$  fuzzy control design for the uncertain fuzzy system \cite{uncertain1,uncertain2,uncertain3,uncertain4,Uncertain_fuzzy}
\begin{equation}\label{eq_m2}
    \begin{split}
   \dot{ \bm{x}}(t) { }={ }& \sum_{i=1}^{\mathcal{N}_r}  \xi_i(\bm{\nu}(t)) \{  [ \bm{A}_i+\Delta \bm{A}_i] \bm{x}(t)+ \bm{B}_{1_i}\bm{w}(t)    +\bm{B}_{2_i}\bm{u}(t) \},   \bm{x}(0)=0 \\
     \bm{y}(t) { }={ }       &  \sum_{i=1}^{\mathcal{N}_r} \xi_i(\bm{\nu}(t)) \bm{C}_{i} \bm{x}(t)+ \bm{D}_{i}  \bm{u}(t) \\
    \end{split}
\end{equation}
where
\begin{equation}\label{eq_def_Hi}
     \Delta \bm{A}_i= F(\bm{x}(t),t)\bm{H}_{i}.
\end{equation}
In~(\ref{eq_m2}),
$\bm{x}(t)$ represents the
state vector, $\bm{w}(t)$ the external disturbances, and~$\bm{y}(t)$ the system output.
The fuzzy controller
\begin{equation}\label{eq_ut}
\bm{u}(t)  =   \sum_{i=1}^{\mathcal{N}_r}  \xi_i(\bm{\nu}(t))\bm{K}_i\bm{x}(t)
\end{equation}
is considered, where
$\xi_i,i=1,2,...,\mathcal{N}_r$, represent normalized fuzzy bases,
$\bm{\nu}(t)$ is the premise vector, $\mathcal{N}_r$ is the number of fuzzy rules, and
$\bm{K}_i,i=1,2,...,\mathcal{N}_r,$ are the controller gains that need to be designed.
In~(\ref{eq_def_Hi}),
  $\bm{H}_{i},i=1,2,...,\mathcal{N}_r,$ are known matrices that characterize the structure of $\Delta \bm{A}_i$,
  and $F(\bm{x}(t),t)$ models the uncertainty
with
\begin{equation}\label{eq_rho}
   || F(\bm{x}(t),t)  ||\leq \frac{1}{\rho}.
\end{equation}

The following lemma shows the LMI formulation for
the controller design~\cite{Uncertain_fuzzy}.
\begin{lma}\label{thm_uncertain_fuzzy_sys}
The system (\ref{eq_m2}) has an $\mathcal{L}_2$-gain less than or equal to $\gamma$, i.e.,
\begin{equation}\label{eq_H_inf_index}
    \int_0^\infty     \bm{y}(t)^T    \bm{y}(t)dt \leq \gamma^2   \int_0^\infty     \bm{w}(t)^T    \bm{w}(t)dt, \bm{x}(0)=0
\end{equation}
for all
\begin{equation*}
 \int_0^\infty     \bm{w}(t)^T    \bm{w}(t)dt<\infty
\end{equation*}
if there exist a matrix $\bm{Z}$, a positive scalar $\delta$, and matrices $\bm{M}_i,i=1,2,...,\mathcal{N}_r,$
such that
\begin{equation}\label{eq_1_LMI}
    \bm{Z}>0 \mbox{ and } \bm{\Omega}_{ij}+\bm{\Omega}_{ji}< 0, 1\leq i \leq j\leq  \mathcal{N}_r
\end{equation}
where
\begin{equation*}
  \bm{\Omega}_{ij}=
\left[
  \begin{array}{ccc}
    (\bm{A}_i \bm{Z},\star)+(\bm{B}_{2_i} \bm{M}_j,\star) & \star & \star \\
    \tilde{\bm{B}}_{1_i}^T & -\gamma \bm{I} & \star \\
    \tilde{\bm{C}}_{i}\bm{Z}+ \tilde{\bm{D}}_{i}\bm{M}_j& 0 & -\gamma \bm{I} \\
  \end{array}
\right]
\end{equation*}
with
\begin{equation}\label{eq_delta_def}
       \tilde{\bm{B}}_{1_i}{ }={ }
       \left[
                                       \begin{array}{cc}
                                         \delta\bm{I} & \bm{B}_{1_i} \\
                                       \end{array}
                                     \right],
       \tilde{\bm{C}}_{i}= \left[
                            \begin{array}{cc}
                              \frac{\gamma }{\rho \delta}  \bm{H}_{i}^T &  \sqrt{2} \bm{C}_{i}^T\\
                            \end{array}
                          \right]^T,\mbox{ and }
 \tilde{\bm{D}}_{i}{ }={ }
  \left[
                            \begin{array}{cc}
                           0&     \sqrt{2}  \bm{D}_{i}^T\\
                            \end{array}
                          \right]^T.
\end{equation}
The controller gains $\bm{K}_i,i=1,2,...,\mathcal{N}_r,$
can be recovered by
$\bm{K}_i=\bm{M}_i \bm{Z}^{-1}$.
\end{lma}
\begin{proof}
Referring  to  the proof of Theorem 3.1 in~\cite{Uncertain_fuzzy},
the result follows from using the augmented disturbance
\begin{equation*}
    \bm{\tilde{w}}(t)=
    \left[
      \begin{array}{cc}
        (\frac{1}{\delta}F(\bm{x}(t),t) \bm{H}_{i} \bm{x}(t))^T & \bm{w}(t)^T \\
      \end{array}
    \right]^T.
\end{equation*}
\end{proof}

In~\cite{Uncertain_fuzzy}, the values of $\gamma$,  $\rho$, and $\delta$ have been heuristically prescribed,
and a
feasibility problem consisting of the constraints~(\ref{eq_1_LMI})
has been dealt with.
For an SO design, we can minimise the attenuation level $\gamma$, and assign the values of the uncertainty tolerance $1/\rho$ in~(\ref{eq_rho}) and the
auxiliary variable $\delta$ in~(\ref{eq_delta_def}) in advance,
yielding an LMIP.
 These conventional design approaches lead to solving a feasibility problem or an SOP.

Suppose that we are interested in the
 controller gains so that the resulting system can tolerate as much uncertainty as possible and achieve
a minimal attenuation level simultaneously. In this scenario,
 it is desired to maximise $1/\rho$  and minimise $\gamma$.
Referring to~(\ref{eq_MOMIP2}), this can be achieved by choosing
\begin{equation}\label{eq_1_content}
\bm{\alpha}=
\left[
  \begin{array}{ccc}
     \gamma  & \rho & \delta \\
  \end{array}
\right]^T,
\bm{X}^*(\bm{\alpha}){ }={ } (\bm{Z}^*(\bm{\alpha}),\bm{M}_1^*(\bm{\alpha}),...,\bm{M}_{\mathcal{N}_r}^*(\bm{\alpha})), \mbox{ and }
 \bm{f}(\bm{\alpha}){ }={ }
\left[
  \begin{array}{ccc}
     \gamma  & \rho  \\
  \end{array}
\right]^T.
\end{equation}
Based on Theorem~\ref{thm_equi} and Lemma~\ref{thm_uncertain_fuzzy_sys},
the controller gains
\begin{equation*}
   \bm{K}_i^*(\bm{\alpha}){ }={ }  \bm{M}_i^*(\bm{\alpha}) \times [ \bm{Z}^*(\bm{\alpha}) ]^{-1},i=1,2,...,\mathcal{N}_r,
\end{equation*}
can be determined by solving
\begin{equation}\label{eq_ex1_MOMIP}
 \begin{split}
  \min_{ \gamma  , \rho , \delta, } \;&
\left[
  \begin{array}{cc}
     \gamma  & \rho  \\
  \end{array}
\right]^T
     \\
 \mbox{subject to }   &  \bm{Z}^*(\bm{\alpha})>0,\bm{\Omega}_{ij}^*(\bm{\alpha})+\bm{\Omega}_{ji}^*(\bm{\alpha})< 0,\\
  &     1\leq i \leq j\leq \mathcal{N}_r.
 \end{split}
\end{equation}
Once the MOMIP~(\ref{eq_ex1_MOMIP}) has been solved, we are able to jointly consider
 the tolerance of system uncertainty and the $H_\infty$ attenuation level.
  Furthermore, all the values of $\gamma$, $\rho$, and $\delta$ are determined through optimisation rather than a heuristic assignment.
 These differentiate our MO design approach from conventional ones.

Another advantage of using the proposed methodology is its flexibility to adjust
 controller gains.
Typically, when an $H_\infty$ design is adopted,
``large values'' of controller gains $\bm{K}_i,i=1,2,...,\mathcal{N}_r$, can result from using a very small value of the $H_\infty$ attenuation level $\gamma$.
However, too large values of $\bm{K}_i$ can be impractical in the implementation.
A trial-and-error method that assigns the attenuation level
is  used mostly to avoid this difficulty.
In our MO design approach, we can modify the objective function to reduce the values of controller gains.
In this example,
it is noted that
\begin{equation*}
 \bm{K}_i^*(\bm{\alpha})=\bm{M}_i^*(\bm{\alpha}) [\bm{Z}^*(\bm{\alpha})]^{-1}= \bm{M}_i^*(\bm{\alpha}) \times \frac{adj(\bm{Z}^*(\bm{\alpha}))}{det(\bm{Z}^*(\bm{\alpha}))}
\end{equation*}
where $adj(\cdot)$ and $det(\cdot)$ denote the classical adjoint and the determinant of a square matrix, respectively.
Roughly speaking,
 we may reduce the value of $1/det(\bm{Z}^*(\bm{\alpha}))$
to  avoid obtaining large $\bm{K}_i^*(\bm{\alpha})$.
Therefore, the 3rd objective $1/det(\bm{Z}^*(\bm{\alpha}))$ is added to the original objective function, yielding the  MOMIP
\begin{equation}\label{eq_ex1_MOMIP_aug}
 \begin{split}
  \min_{ \gamma  , \rho , \delta, } \;&
\left[
  \begin{array}{ccc}
     \gamma  & \rho &  \frac{1}{det(\bm{Z}^*(\bm{\alpha}))} \\
  \end{array}
\right]^T
     \\
 \mbox{subject to }   &  \bm{Z}^*(\bm{\alpha})>0,\bm{\Omega}_{ij}^*(\bm{\alpha})+\bm{\Omega}_{ji}^*(\bm{\alpha})< 0,\\
  &  1\leq i \leq j\leq \mathcal{N}_r.
 \end{split}
\end{equation}
Our simulations will demonstrate the effectiveness of this method.

\subsection{Example 2: Bounded-Input Bounded-Output (BIBO) System Design}\label{subsec_output_control}

In real-world control problems, bounded outputs and inputs are usually desired.
Generally speaking, lowering the upper bound on the input and lowering that on the output  are two conflicting objectives.
By lowering the upper bound on the control input, less input energy is used.
As a result,
the output performance can deteriorate, increasing the upper bound on the output.

Consider a linear system
\begin{equation}\label{eq_m3}
    \begin{split}
   \dot{ \bm{x}}(t) { }={ }& \bm{A} \bm{x}(t)+ \bm{B}\bm{u}(t)  \\
     \bm{y}(t) { }={ }    &  \bm{C} \bm{x}(t).\\
    \end{split}
\end{equation}
The following lemma provides LMI conditions for the bounds on the input~$\bm{u}(t)=  \bm{K}\bm{x}(t)$  and output~$ \bm{y}(t)$ of the control system~(\ref{eq_m3}).
\begin{lma}\label{thm_out_control}
If there exist matrices $\bm{Z}$ and $\bm{M}$ such that
\begin{equation}\label{eq_out_control}
   (\bm{A}  \bm{Z},\star)+ (\bm{B}  \bm{M},\star) <0  ,
\bm{Z}_1=
\left[
   \begin{array}{cc}
     1 & \star \\
     \bm{x}(0) & \bm{Z} \\
   \end{array}
 \right]>0,
 \bm{Z}_2=
\left[
   \begin{array}{cc}
    \bm{Z} & \star \\
     \bm{M} & \bar{u}^2 \bm{I} \\
   \end{array}
 \right]>0,\mbox{and }
\bm{Z}_3=
\left[
   \begin{array}{cc}
    \bm{Z} & \star \\
     \bm{C}  \bm{Z} & \bar{y}^2 \bm{I} \\
   \end{array}
 \right]>0
\end{equation}
then the system~(\ref{eq_m3}) is quadratically stabilizable
with
\begin{equation*}
|| \bm{u}(t) || = ||   \bm{K}\bm{x}(t)  ||<\bar{u} \mbox{ and }       || \bm{y}(t) || < \bar{y}
\end{equation*}
 using
the controller gain $\bm{K}=\bm{M} \bm{Z}^{-1}$.
\end{lma}
\begin{proof}
The reader can refer to Theorems~11 and 12 in~\cite{Tanaka_fuzzy} or (7.16) in ~\cite{Boyd_LMI} for a detailed proof.
\end{proof}

Referring to~(\ref{eq_MOMIP2}), we let
\begin{equation*}
 \bm{f}(\bm{\alpha})=\bm{\alpha}=
 \left[
   \begin{array}{cc}
   \bar{u} & \bar{y} \\
   \end{array}
 \right]^T
 \mbox{ and }
\bm{X}^*(\bm{\alpha}){ }={ } (\bm{Z}^*(\bm{\alpha}),\bm{M}^*(\bm{\alpha})).
\end{equation*}
Based on Theorem~\ref{thm_equi} and Lemma~\ref{thm_out_control},
the MOMIP associated with the BIBO system design can be formulated as
\begin{equation}\label{eq_ex3_MOMIP}
 \begin{split}
  \min_{ \bm{\alpha}} \; &     \bm{\alpha}
     \\
  \mbox{subject to }  &   \bm{Z}_1^*(\bm{\alpha}) >0  , \bm{Z}_2^*(\bm{\alpha}) >0  , \bm{Z}_3^*(\bm{\alpha}) >0,\\
  &      (\bm{A}  \bm{Z}^*(\bm{\alpha}),\star)+ (\bm{B}  \bm{M}^*(\bm{\alpha}),\star) <0
   \end{split}
\end{equation}

Conventionally, both
bounds were assigned in advance, or one bound was minimised with the other bound prescribed.
In contrast, our MO approach
 minimises the upper bounds simultaneously, as shown in~(\ref{eq_ex3_MOMIP}).
 This provides the designer with a broader perspective on optimality, and on Pareto optimality in particular.

Due to space limitation, only two design examples have been  examined.
The proposed methodology can be extended to other optimisation problems with
LMI constraints. More examples with possible MO formulations in the form of~(\ref{eq_MOMIP})
are discussed as follows. In~\cite{DA1} and~\cite{DA2}, the problem of computing the region of attraction (ROA) was investigated.
Suppose that the estimation of the ROA of a system has been formulated as an SOP with LMI constraints and that a robust control law is desired to stabilize
the system.
To apply the proposed MO approach to this scenario, we may maximise the estimated ROA and minimise an $H_\infty$ attenuation level.
In~\cite{anti_wind}, the design of dynamic anti-windup compensators was addressed, and design problems having the same LMI constraints but different objectives were formulated as separate LMIPs.
Among them, one design aimed at the maximisation of the disturbance tolerance, and another design focused on the maximisation of the disturbance attenuation.
For an MO extension, we may consider an MOP in which  the disturbance tolerance and the disturbance attenuation are maximised.

In~\cite{SPR} and~\cite{UnC}, strictly positive real (SPR) controllers and
control systems with ellipsoidal parametric uncertainty were considered, respectively.
LMI approaches were employed for controller synthesis in both studies.
An $H_2$ design that addresses the average system performance was adopted in~\cite{SPR},
and  an  $H_\infty$ design that focuses on
the robustness of a system in
 worst-case scenarios  was used in~\cite{UnC}.
For a mixed $H_2/H_\infty$ design and its MO extension, both $H_2$ and $H_\infty$ performance indices can be minimised simultaneously in our MO approach.
Finally, in~\cite{MPC}, model predictive control (MPC) of nonlinear systems was considered.
For an MO formulation associated with the  MPC problem,
minimising the upper bound on the infinite horizon cost and
 the upper bound on the input energy can be the two objectives.
In this case, further research is needed to address the computational complexity
for an online application.

\section{Proposed HMODE Algorithm}\label{sec_HMODE}

As mentioned previously, there are several
advantages of solving the MOMIP~(\ref{eq_MOMIP}) or its equivalent form~(\ref{eq_MOMIP2})
for an MO controller design.
In this section, we develop an algorithm that can solve~(\ref{eq_MOMIP2}).
Since DE algorithms have been found to be very robust and applied to a large number of SOPs,
their basic structure is used in the proposed algorithm.
To manage the MI constraints, we integrate interior-point methods into our solution searching scheme, resulting in the hybrid algorithm.
Roughly speaking, infeasible solutions ($\lambda^*(\bm{\alpha})>0$, as shown in the proof of Lemma~\ref{lma_1}) or dominated solutions tend to be removed during the algorithm iteration.
Feasible and nondominated solutions are thus remained  so that an APF can be obtained when the algorithm stops.

To facilitate the ensuing discussion,
let $rand_{(a,b)}$ and
$unidrnd_M$ denote
 continuous and discrete uniform random variables over $(a,b)$
and  $\{1,2,...,M \}$, respectively.
For column vectors $\bm{a}$ and $\bm{b}\in \mathbb{R}^N$,
the expression $\bm{a} < \bm{b}$ ($\bm{a} \leq \bm{b}$)  stands for
$[\bm{a}]_i < [\bm{b}]_i$ ($[\bm{a}]_i \leq [\bm{b}]_i$) for $i=1,2,...,N$, where
$[\cdot]_i$ denotes the $i$th entry of a vector.
Furthermore,  $\bm{a} \in [ \bm{\underline{a}},\bm{\bar{a}} ]$ implies that
$\bm{\underline{a}} \leq \bm{a}$ and   $\bm{a} \leq \bm{\bar{a}}$.
For a matrix $\bm{A}$, $[\bm{A}]_{ij}$ denotes
 its $(i,j)$th entry.
We are now in the position to present
the proposed algorithm as follows.\\
\textit{\textbf{HMODE Algorithm}} \\
\textbf{Input}:
\begin{itemize}
  \item  MOMIP~(\ref{eq_MOMIP2}).
  \item $\mathcal{N}_{p}$, the population size;
  \item $\mathcal{N_I}$, the number of iterations;
  \item $\eta_c\in (0,1)$, control parameter of the crossover;
    \item $\eta_d$, control parameter of the crowding distance;
  \item $[\bm{\underline{\alpha}},\bm{\bar{\alpha}} ]$, the range of interest for the non-matrix decision variable $\bm{\alpha}$.
\end{itemize}
\textbf{Step 1) Initialization}:
      \begin{itemize}
          \item []\textbf{Step 1.1)} Let $\Gamma_f=\emptyset$, the set consisting of the best-so-far vectors in the objective function space;
          $\Gamma=\emptyset$, the set consisting of matrix and non-matrix variables in the decision variable space (each point
          in $\Gamma$ corresponds to a vector stored in  $\Gamma_f$); and the algorithm counter $t_c=1$.
          \item []\textbf{Step 1.2)}
          Randomly generate an initial population
          \begin{equation*}
           \bm{\alpha}_1,\bm{\alpha}_2,...,\bm{\alpha}_{\mathcal{N}_{p}} \in [\bm{\underline{\alpha}},\bm{\bar{\alpha}} ].
          \end{equation*}
      \end{itemize}
\textbf{Step 2) Differential Evolution}:\\
\textbf{For} $i=1$ to $\mathcal{N}_{p}$ \textbf{do}

\textbf{If} the Phase-I stage is active, e.g., $t_c\leq \frac{2}{3}\mathcal{N_I}$, \textbf{then}
\begin{itemize}
  \item[] \textbf{Step 2.I.1)} Mutation operation:
      \begin{equation}\label{eq_mutate}
        \bm{v}_i= \bm{\alpha}_{best}+ rand_{(0,1)}(i) ( \bm{\alpha}_j-\bm{\alpha}_k  )
      \end{equation}
     where $j$ and $k$ are distinct integers and randomly selected from $\mathbb{Z}_+\cap [1,\mathcal{N}_{p}]\setminus \{i\}$,
      and $\bm{\alpha}_{best}$ is randomly chosen from $\Gamma$ if $\Gamma\not=\emptyset$. If $\Gamma=\emptyset$, then $\bm{\alpha}_{best}=\bm{\alpha}_{\ell}$ for some integer $\ell$ randomly selected from $\mathbb{Z}_+\cap [1,\mathcal{N}_{p}]\setminus \{i,j,k\}$.
  \item[]   \textbf{Step 2.I.2)} Reflection operation:
\begin{equation}\label{eq_reflect}
    [\bm{v}_{i}']_j=
    \left\{
      \begin{array}{ll}
        \min \{  [\bm{\bar{\alpha}}]_j ,  2  [\bm{\underline{\alpha}}]_j - [\bm{v}_{i}]_j  \} , & \hbox{if }  [\bm{v}_i]_j<  [\bm{\underline{\alpha}}]_j  \\
        \max \{ [\bm{\underline{\alpha}}]_j  ,  2 [\bm{\bar{\alpha}}]_j  -[\bm{v}_{i}]_j  \}, & \hbox{if } [\bm{v}_i]_j >  [\bm{\bar{\alpha}}]_j\\
      {[}\bm{v}_{i}]_j, & \hbox{otherwise.}
      \end{array}
    \right.
\end{equation}
\item[]   \textbf{Step 2.I.3)} Crossover operation:
\begin{equation}\label{eq_crossover}
  [\bm{\varphi}_i]_j=
    \left\{
      \begin{array}{ll}
        [\bm{v}_i']_j, & \hbox{if } rand_{(0,1)}(i,j) \leq \eta_c \hbox{ or }  \\
                    &    j= unidrnd_M(i,j) \\
         {[} \bm{\alpha}_i]_j, & \hbox{otherwise.}
      \end{array}
    \right..
\end{equation}
\item[]   \textbf{Step 2.I.4)} Selection operation: \\
 Assign
\begin{equation*}
\bm{\alpha}_i= \bm{\varphi}_{i}
\end{equation*}
if either
\begin{equation}\label{eq_C1}
\lambda^*( \bm{\varphi}_{i})<0 <  \lambda^*( \bm{\alpha}_{i})
\end{equation}
or
\begin{equation}\label{eq_C2}
\lambda^*( \bm{\varphi}_{i}),\lambda^*( \bm{\alpha}_{i})<0 \mbox{ and }
 \bm{f}(\bm{\varphi}_{i})    < \bm{f}(\bm{\alpha}_{i})
\end{equation}
holds true.
\end{itemize}

\textbf{Else}  (Phase-II stage is active, e.g., $t_c > \frac{2}{3}\mathcal{N_I}$)
\begin{itemize}
  \item[] Assign
      \begin{equation}\label{eq_add_crossover}
        \bm{\alpha}_i= \bm{R}_{i} \bm{\alpha}_i' +(\bm{I}_M-\bm{R}_{i}) \bm{\alpha}_i''
      \end{equation}
where $\bm{I}_M$  represents the $M\times M$ identity matrix,
$\bm{\alpha}_i'$ and $\bm{\alpha}_i''$ are distinct vectors randomly selected from $\Gamma$, and
$\bm{R}_{i} \in \mathbb{R}^{M \times M}$ is defined as
\begin{equation}\label{eq_def_R}
    [\bm{R}_{i}]_{jk}=
\left\{
  \begin{array}{ll}
    rand_{(0,1)}(i,j), & \hbox{if }   1\leq j=k \leq M  \\
    0, & \hbox{otherwise.}
  \end{array}
\right.
\end{equation}
\end{itemize}

\textbf{End If}\\
\textbf{End For}\\
\textbf{Step 3) Update:} Let
\begin{equation*}
    \vartheta(1), \vartheta(2), ..., \vartheta(\kappa)
\end{equation*}
be the indices of all non-matrix decision variables that are feasible, i.e.,
\begin{equation}\label{eq_feasible_update}
  \lambda^*( \bm{\alpha}_{\vartheta(i)} ) <0, i=1,2,...,\kappa.
\end{equation}
 \textbf{For}  $i=1$ to $\kappa$ \textbf{do}

 \textbf{If} $\bm{f}(\bm{\alpha}_{\vartheta(i)}  )$ is not dominated by any vectors in $\Gamma_f$, \textbf{then}
       \begin{itemize}
       \item [] \textbf{Step 3.1)} Remove all
        \begin{equation*}
         \bm{f}(\bm{\alpha})\in \Gamma_f  \mbox{ and }  (\bm{\alpha}, \bm{X}^*(\bm{\alpha})) \in \Gamma
        \end{equation*}
         if  $\bm{f}(\bm{\alpha}_{\vartheta(i)})$ dominates $\bm{f}(\bm{\alpha})$.
        \item [] \textbf{Step 3.2)} Add $\bm{f}(\bm{\alpha}_{\vartheta(i)})$ to $\Gamma_f$ and the associated design parameters $(\bm{\alpha}_{\vartheta(i)}, \bm{X}^*(\bm{\alpha}_{\vartheta(i)}))$ to
$\Gamma$ if
\begin{equation}\label{eq_dis_check}
||\bm{f}(\bm{\alpha}_{\vartheta(i)})- \Gamma_f ||> \eta_d
\end{equation}
where $||\bm{f}(\bm{\alpha}_{\vartheta(i)})- \Gamma_f ||$ represents the Euclidean distance between
the vector $\bm{f}(\bm{\alpha}_{\vartheta(i)})$ and the set $\Gamma_f$.
       \end{itemize}

        \textbf{End If}\\
\textbf{End For}\\
\textbf{Step 4) Stopping Criterion}: Set $t_c=t_c+1$.  If $t_c\leq \mathcal{N_I}$, then go to \textbf{Step 2)}. Otherwise, go to
\textbf{Step 5)}.\\
\textbf{Step 5) Knee Selection}:
Evaluate
\begin{equation}\label{eq_knee}
 \bm{\alpha}^*=\arg_{ \bm{\alpha}} \; \max_{ \bm{f}( \bm{\alpha})\in \Gamma_f }    \;   \prod_{n=1}^N  \frac{\bar{f}_n -  [ \bm{f}(\bm{\alpha})]_n }{\bar{f}_n -\underline{f}_n}
\end{equation}
where
\begin{equation*}
    \bar{f}_n = \max_{\bm{f}(\bm{\alpha})\in \Gamma_f }  [ \bm{f}(\bm{\alpha})]_n \mbox{ and }
   \underline{f}_n = \min_{\bm{f}(\bm{\alpha})\in \Gamma_f }  [ \bm{f}(\bm{\alpha})]_n
\end{equation*}
for  $n=1,2,...,N$.\\
\textbf{Output}:
\begin{itemize}
  \item $\Gamma_f$, the APF.
  \item $\bm{f}(\bm{\alpha}^*)$, the knee selected from $\Gamma_f$.
  \item $(\bm{\alpha}^*,\bm{X}^*(\bm{\alpha}^*))$, the design parameters selected from $\Gamma$.
\end{itemize}

Several points regarding the proposed algorithm are addressed in the following paragraphs.

For the inputs of the proposed HMODE algorithm,
 we often have
 $ \bm{\underline{\alpha}} >0$
for the range of interest~$[\bm{\underline{\alpha}},\bm{\bar{\alpha}} ]$
   as each entry of the vector $\bm{\alpha}$ mostly represents a certain performance index.
In Step~1.1),  the sets $\Gamma_f$ and $\Gamma$
have been used to store best-so-far vectors and the corresponding design parameters, respectively.
The structures of these sets are clear from Step~3.1), i.e.,
 \begin{equation*}
     \Gamma_f \subset  \{  \bm{f}(\bm{\alpha}): \bm{\alpha}\in [\bm{\underline{\alpha}},\bm{\bar{\alpha}} ]  \} \mbox{ and }
         \Gamma =    \{  (\bm{\alpha}, \bm{X}^*(\bm{\alpha})) :\bm{f}(\bm{\alpha}) \in \Gamma_f \}.
   \end{equation*}

It is well-known that DE algorithms have a fast convergence rate but
 they often reach the vicinity of the true Pareto front~\cite{draw_back_DE,Coello}.
 To overcome this problem and maintain the convergence rate,
a two-phase scheme has been proposed in~\cite{two-phase}.
Motivated by the idea of using a two-phase scheme,
 we have adopted the basic structure of conventional DE algorithms at our Phase-I stage,
  and performed a crossover operation~(\ref{eq_add_crossover}) over best-so-far solutions at the Phase-II stage.
The Phase-I and Phase-II stages are used to address the exploration and exploitation of a search space, respectively.
These stages are key to finding feasible and nondominated solutions.

In Step~2.I.1), there are several ways to calculate the mutant, e.g.,
``DE/rand/1'', ``DE/best/1'', ``DE/current-to-best/1'', and ``DE/current-to-rand/1''~\cite{mutate_way1,mutate_way2}.
In general, $\bm{\alpha}_{best}$ in~(\ref{eq_mutate}) should yield good performance in the objective function space
 to achieve fast convergence in the mutating process. For an SOP, ``good performance'' means that $\bm{\alpha}_{best}$ can achieve
a small objective value.
 In the proposed algorithm, solutions in $\Gamma$ corresponding to best-so-far  vectors in $\Gamma_f$
  serve as
 $\bm{\alpha}_{best}$.

To ensure that  mutants satisfy the range of interest $[\bm{\underline{\alpha}},\bm{\bar{\alpha}} ]$, the modification~(\ref{eq_reflect}) in Step~2.I.2)
has been performed~\cite{reflec_op}. This is termed the reflection operation
because  $[\bm{v}_i]_j$ is reflected to $[\bm{v}'_i]_j$   about $[\bm{\underline{\alpha}}]_j$
if $[\bm{v}_i]_j< [\bm{\underline{\alpha}}]_j $
and  about $[\bm{\bar{\alpha}}]_j$
if $[\bm{v}_i]_j> [\bm{\bar{\alpha}}]_j $.
It is evident from~(\ref{eq_mutate}) and~(\ref{eq_reflect})  that the mutant $\bm{v}_i'$ satisfies
$\bm{v}_i' \in [\bm{\underline{\alpha}},\bm{\bar{\alpha}} ]$.
Similarly,
the offspring $\bm{\varphi}_i$ produced from the crossover operation~(\ref{eq_crossover}) in Step~2.I.3)
satisfies
$  \bm{\varphi}_i  \in[\bm{\underline{\alpha}},\bm{\bar{\alpha}} ]$.
In~Step~2.I.4),  the offspring $\bm{\varphi}_i$ is selected if it is better than its predecessor, i.e., either~(\ref{eq_C1}) or~(\ref{eq_C2}) holds true, each of which requires that  $\bm{\varphi}_i$ be feasible, i.e., $\lambda^*(\bm{\varphi}_i)<0$.
By using the selection operation defined in~Step~2.I.4),
conventional DE algorithms for SO optimisation have been extended to address the MOMIP.

\begin{figure}
\centering
  \includegraphics[width=4cm]{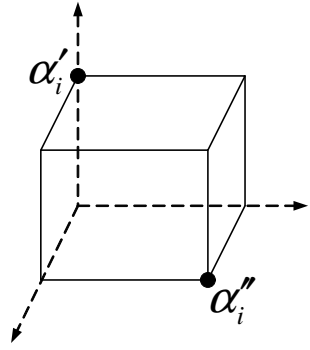}\\
  \caption{Illustration of the cube in~(\ref{eq_cubic}) for $M=3$. In this case, the offspring $\bm{\alpha}_i$ of
$\bm{\alpha}_i'$ and $\bm{\alpha}_i''$
  produced by~(\ref{eq_add_crossover}) is a point within this cube.}\label{fig_cubic}
\end{figure}

At the Phase-II stage,
we explore the neighborhood of solutions in $\Gamma$ using~(\ref{eq_add_crossover})
 in order to
 reach the true Pareto front.
Due to the design of $[\bm{R}_i]_{jk}$ in~(\ref{eq_def_R}),
the newly produced offspring $\bm{\alpha}_i$ in~(\ref{eq_add_crossover})
can be regarded as a point randomly chosen from the cube
\begin{equation}\label{eq_cubic}
    \prod_{j=1}^M   [  \min \{   [\bm{\alpha}_i']_j, [\bm{\alpha}_i'']_j       \}, \max\{   [\bm{\alpha}_i']_j, [\bm{\alpha}_i'']_j         \}               ].
\end{equation}
Fig.~\ref{fig_cubic} shows the graphical illustration of the cube defined in~(\ref{eq_cubic})
for $M=3$.
Analogous to the property that $\bm{v}_i'$ and $\bm{\varphi}_i$ are always confined within $[\bm{\underline{\alpha}},\bm{\bar{\alpha}} ]$, we have  $\bm{\alpha}_i\in [\bm{\underline{\alpha}},\bm{\bar{\alpha}} ]$ for all $i$ at the Phase-II stage.
We conclude that
\begin{equation*}
\bm{v}_i'{ }={ }
\bm{v}_i'(t_c)\in [\bm{\underline{\alpha}},\bm{\bar{\alpha}} ],  \bm{\varphi}_i=\bm{\varphi}_i(t_c)\in [\bm{\underline{\alpha}},\bm{\bar{\alpha}} ], \mbox{ and }
   \bm{\alpha}_i{ }={ }
   \bm{\alpha}_i(t_c)
 \in [\bm{\underline{\alpha}},\bm{\bar{\alpha}} ]
\end{equation*}
for all $i$ and $t_c$.

\begin{figure}
\centering
  \includegraphics[width=6cm]{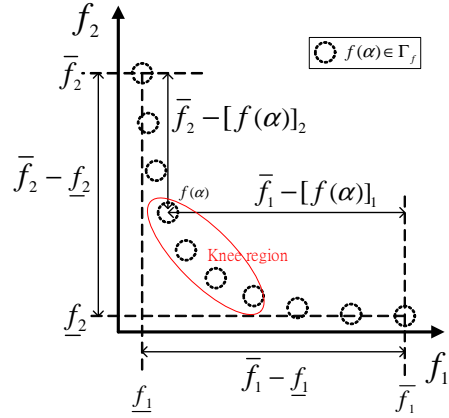}\\
  \caption{Diagram of (\ref{eq_knee}) with $N=2$. For the $n$th objective, the value $\bar{f}_n -  [\bm{f}(\bm{\alpha})]_n $
  represents the improvement achieved by the solution $\bm{\alpha}$, and  the value  $\bar{f}_n -\underline{f}_n$
  represents the maximum improvement with respect to  $\Gamma_f$.}\label{fig_knee_illu}
\end{figure}

In Step~3), the update process is active only if feasible solutions have been found, i.e., the conditions~(\ref{eq_feasible_update})  hold true for a positive integer~$\kappa$. To maintain a manageable size of the APF, a simple mechanism~(\ref{eq_dis_check}) has been introduced: if $\bm{f}(\bm{\alpha})$ is not dominated by
any vectors in~$\Gamma_f$
and poses a distance away from the set~$\Gamma_f$, then
$\bm{f}(\bm{\alpha})$ and $(\bm{\alpha},\bm{X}^*(\bm{\alpha}))$
 enter the external archives $\Gamma_f$ and $\Gamma$, respectively.
In Step~4),
once $t_c$, the counter of the number of iterations, has exceeded the prescribed $\mathcal{N}_\mathcal{I}$,
the DE in Step~2) will no longer be executed, and
 a set of design parameters becomes available.
The aim of the remaining work is to choose the most reasonable design for the control system.

In Step~5), the final design has been selected by the optimisation process~(\ref{eq_knee}) that fully exploits the knowledge extracted from the APF $\Gamma_f$.
Since knee solutions are often preferred,
 the chosen design should be able to reflect this preference.
To illustrate the underlying idea of using~(\ref{eq_knee}),
Fig.~\ref{fig_knee_illu} presents a case in which two objectives are involved. In such a case, the quantities
\begin{equation}\label{eq_two_quan}
\frac{\bar{f}_1 -  [\bm{f}(\bm{\alpha})]_1 }{\bar{f}_1 -\underline{f}_1}  \mbox{ and } \frac{\bar{f}_2 -  [\bm{f}(\bm{\alpha})]_2 }{\bar{f}_2 -\underline{f}_2}
\end{equation}
 can be interpreted as the normalized improvement with respect to the first and second objectives, respectively.
According to~(\ref{eq_knee}), the design that maximises the ``overall improvement'' will be selected, in which the overall improvement is defined
as the the product of the separate improvement of each objective, i.e., the product of the two terms in~(\ref{eq_two_quan}).
To some extent, the optimisation process~(\ref{eq_knee}) aims to create a win-win situation
by searching for the ``optimal'' design that sacrifices each objective to some degree in order to do well on the whole.

In the end of the proposed HMODE algorithm, we output valuable information $\Gamma_f$ and  $\bm{f}(\bm{\alpha}^*)$, allowing
the designer of control systems to  have a good understanding of the chosen design parameters $(\bm{\alpha}^*,\bm{X}^*(\bm{\alpha}^*))$.

\section{Numerical Examples}\label{sec_sim}

This section provides a simulation study of the proposed methodology.
The values $\mathcal{N}_{p}=100,\mathcal{N_I}=200, \eta_c=0.2,$ and $\eta_d=0.05$
have been chosen as the inputs to the HMODE algorithm.
The proposed algorithm has been applied to solving the MOMIPs in the design examples in Section~\ref{sec_example}.

\subsection{Example 1: Robust Control Design}

Consider the chaotic Lorenz system~\cite{uncertain4,Uncertain_fuzzy}
\begin{equation}\label{eq_Lorenz}
\begin{split}
  \dot{x}_1(t) { }={ } & -\sigma x_1(t)+ \sigma x_2(t)+ u(t)+0.1 w_1(t) \\
  \dot{x}_2(t) { }={ } & r x_1(t)- x_2(t) - x_1(t)x_3(t)  +0.1 w_2(t) \\
    \dot{x}_3(t) { }={ } & x_1(t)x_2(t) -b x_3(t)     +0.1 w_3(t) \\
    \bm{y}(t) { }={ } &
    \left[
      \begin{array}{ccc}
   x_1(t)+u(t) & x_2(t)+u(t) & x_3(t)+u(t) \\
      \end{array}
    \right]^T
\end{split}
\end{equation}
where $\sigma,r,$ and $b$ are uncertain parameters with
the nominal values
\begin{equation}\label{eq_nominal}
\sigma \approx 10, r \approx28, \mbox{ and } b\approx 8/3
\end{equation}
 which are known to the designer.
The nonlinear system~(\ref{eq_Lorenz}) has been interpolated by a fuzzy system in~(\ref{eq_m2})
with $\mathcal{N}_r=2$ and~\cite{Uncertain_fuzzy}
\begin{equation}\label{eq_ex1_fuzzy}
{\tiny
    \begin{split}
  \xi_1(x_1(t)){ }={ }&  \max\{\; \min\{\;-\frac{1}{50}\times(x_1(t)-30), 1\}, 0\;\}, \xi_2(x_1(t)){ }={ } \max\{\;  \min\{\;\frac{1}{50}\times(x_1(t)+20),  1\}, 0\;\}\\
      \bm{A}_1  { }={ } &
     \left[
       \begin{array}{ccc}
         -10 & 10 & 0 \\
         28 & -1 & 20 \\
         0 & -20 & -8/3 \\
       \end{array}
     \right], \bm{A}_2 =
     \left[
       \begin{array}{ccc}
         -10 & 10 & 0 \\
         28 & -1 & -30 \\
         0 & 30 & -8/3 \\
       \end{array}
     \right],
   \bm{B}_{1_1}  { }={ }    \bm{B}_{1_2}=
        \left[
       \begin{array}{ccc}
         0.1 & 0 & 0 \\
         0 & 0.1 & 0 \\
         0 & 0 & 0.1 \\
       \end{array}
     \right],
   \bm{B}_{2_1}  =   \bm{B}_{2_2}=
        \left[
       \begin{array}{c}
         1  \\
         0  \\
         0  \\
       \end{array}
     \right],   \\
  \bm{C}_{1}  { }={ } &    \bm{C}_{2}=
        \left[
       \begin{array}{ccc}
         1 & 0 & 0 \\
         0 & 1 & 0 \\
         0 & 0 & 1 \\
       \end{array}
     \right],
   \bm{D}_{1}  =   \bm{D}_{2}=
        \left[
       \begin{array}{c}
         1  \\
         1  \\
         1  \\
       \end{array}
     \right], \mbox{and }
   \bm{H}_{1}  { }={ }    \bm{H}_{2}=
        \left[
       \begin{array}{ccc}
    -3  &  3 &        0  \\
    8.4  & 0  &     0   \\
         0  &       0 &  -0.8
       \end{array}
     \right].
    \end{split}
    }
\end{equation}

\begin{figure}
\centering
\begin{equation*}
\begin{array}{cc}
 \includegraphics[width=6cm]{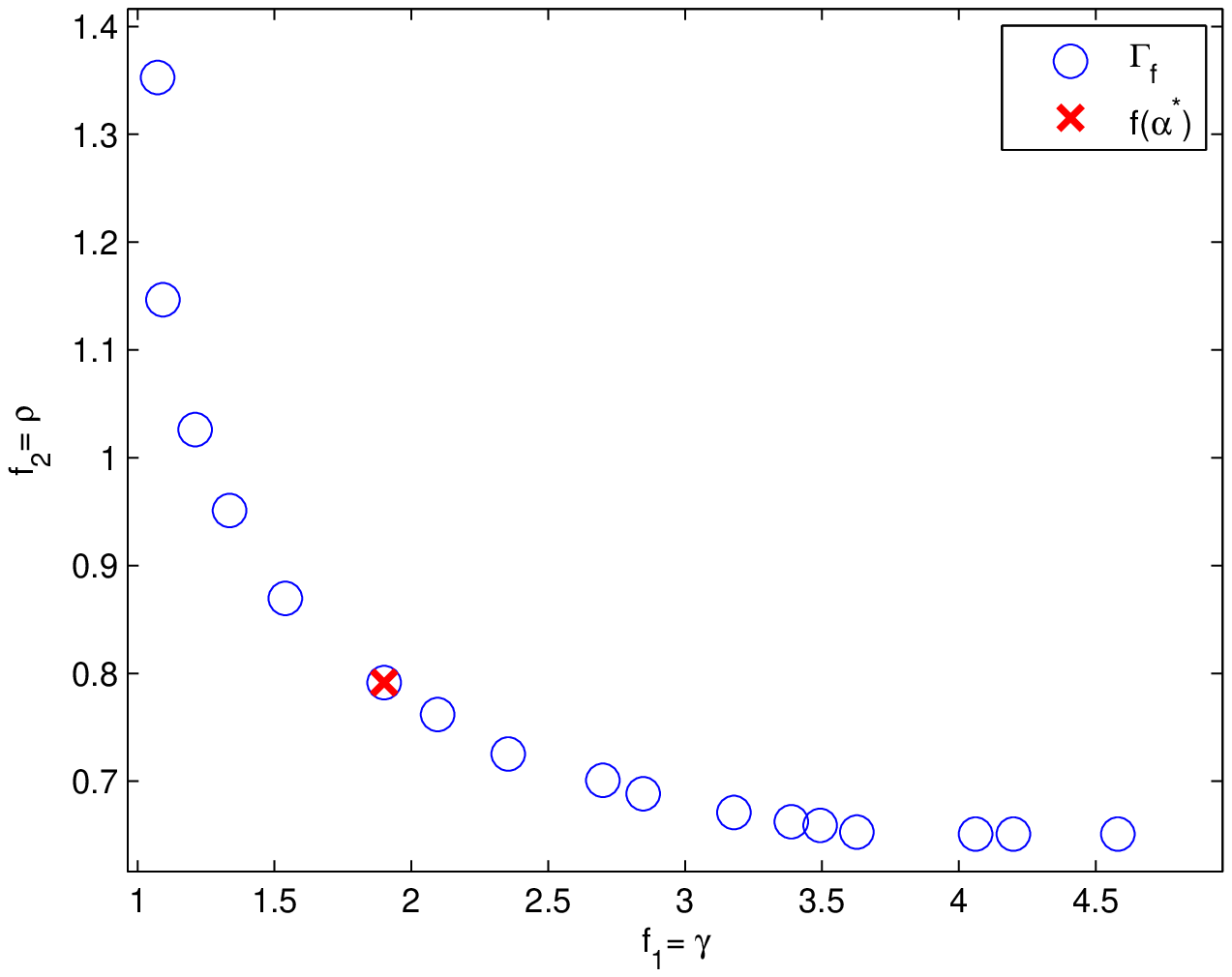} & \includegraphics[width=6cm]{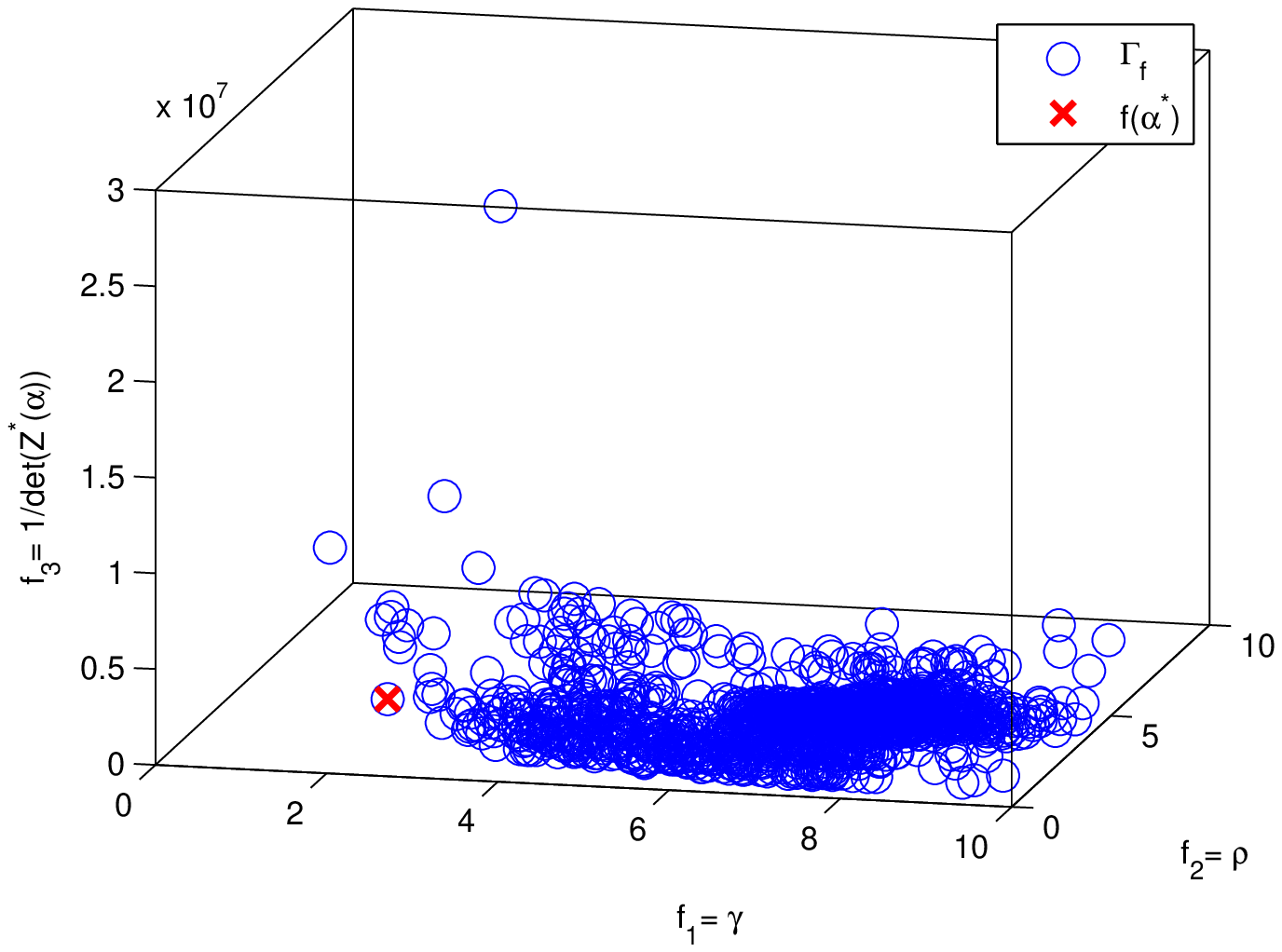} \\
  \mbox{(a)} & \mbox{(b)}
\end{array}
\end{equation*}
\caption{APFs in Example 1: (a) the APF resulting from solving~(\ref{eq_ex1_MOMIP}) and (b) the APF resulting from solving~(\ref{eq_ex1_MOMIP_aug}). Both cases show that the proposed HMODE algorithm can yield a reasonable tradeoff design represented by a knee $\bm{f}(\bm{\alpha}^*)$ of the APF~$\Gamma_f$.
}\label{fig_ex1_PF}
\end{figure}

The range of interest
\begin{equation*}
\bm{\alpha}
\in
    [0.5 ,5] \times  [0.5 ,5] \times  [0.01,5]
\end{equation*}
has been chosen, i.e.,
\begin{equation*}
    \bm{\underline{\alpha}}=
    \left[
                              \begin{array}{ccc}
                                0.5 & 0.5 & 0.01 \\
                              \end{array}
                            \right]^T
   \mbox{ and }
      \bm{\bar{\alpha}}=
    \left[
                              \begin{array}{ccc}
                                5 & 5 & 5 \\
                              \end{array}
                            \right]^T.
\end{equation*}
Solving~(\ref{eq_ex1_MOMIP}) produced
the APF in Fig.~\ref{fig_ex1_PF}(a).
The solution
\begin{equation*}
\bm{\alpha}^*=
\left[
  \begin{array}{ccc}
 1.9009  &   0.7914 &   0.1585 \\
  \end{array}
\right]^T
\end{equation*}
has been selected, which corresponds to the gain matrices
\begin{equation}\label{eq_ex1_K_ori}
\bm{K}_1^*(\bm{\alpha}^*)\approx \bm{K}_2^*(\bm{\alpha}^*)
  \approx{ }
\left[
  \begin{array}{ccc}
-1785 &  -526.5  &  18.3 \\
  \end{array}
\right].
\end{equation}

Several points should be addressed here.
First, while the values $\delta=1$ and $\rho=1$ were heuristically prescribed in~\cite{Uncertain_fuzzy},
they were evaluated through an optimisation process
in our MO approach, which should be better in general.
Second, the shape of the APF in Fig.~\ref{fig_ex1_PF}(a) shows
how the tolerance of uncertainty and the attenuation level affect each other:
the objectives are dependent and conflicting, and
the designer must trade off any improvement of  one objective with the deterioration of the other objective.
As shown in Fig.~\ref{fig_ex1_PF}(a),
the proposed approach can provide a reasonable design that
corresponds to the nondominated vector $\bm{f}(\bm{\alpha}^*)$ lying in the knee region of the APF.
This illustrates the validity of using~(\ref{eq_knee}).
Finally, the resulting gain matrices in~(\ref{eq_ex1_K_ori}) can be too large, which is often undesired in practical applications.

To  avoid obtaining large $\bm{K}_i^*(\bm{\alpha})$, the MOMIP~(\ref{eq_ex1_MOMIP_aug})
has been solved, yielding the new solution
\begin{equation*}
\bm{\alpha}^*=
\left[
  \begin{array}{ccc}
 2.0900  &  2.6961  &  0.1469 \\
  \end{array}
\right]^T
\end{equation*}
and the gain matrices
\begin{equation}\label{eq_ex1_gain}
{\small
\bm{K}_1^*(\bm{\alpha}^*){ }={ }
\left[
  \begin{array}{ccc}
-105.4586 & -34.9714  &  1.7917 \\
  \end{array}
\right], \mbox{ and }
\bm{K}_2^*(\bm{\alpha}^*){ }={ }
\left[
  \begin{array}{ccc}
-105.4506 & -34.9705  &  1.8222\\
  \end{array}
\right].
}
\end{equation}
The values of $\bm{K}_1^*(\bm{\alpha}^*)$ and $\bm{K}_2^*(\bm{\alpha}^*)$ in~(\ref{eq_ex1_gain}) have been reduced  compared with those in~(\ref{eq_ex1_K_ori}).
Fig.~\ref{fig_ex1_PF}(b) shows the corresponding APF.

For the chaotic system~(\ref{eq_Lorenz}), we have chosen
\begin{equation*}
\begin{split}
\bm{x}(0) { }={ } &
\left[
  \begin{array}{ccc}
    x_1(0) & x_2(0) & x_3(0) \\
  \end{array}
\right]^T=
\left[
  \begin{array}{ccc}
  0 & 0 & 0 \\
  \end{array}
\right]^T,
    \sigma{ }={ } \sigma(t)=10+\sin(t),r =  0.8\times 28 , b= 1.1\times 8/3
 \\
        w_1(t){ }={ }  &  rand_{(-0.2,0.8)}(t),w_2(t)=rand_{(-0.7,1)}(t) ,
     w_3(t)   { }={ }   rand_{(-0.1,0.3)} (t)
\end{split}
\end{equation*}
and used the control law
\begin{equation}\label{eq_ut_ex1}
u(t) =  \xi_1(x_1(t))\bm{K}_1^*(\bm{\alpha}^*)\bm{x}(t)+ \xi_2(x_1(t))\bm{K}_2^*(\bm{\alpha}^*) \bm{x}(t)
\end{equation}
 in our simulations.
Figs.~\ref{fig_ex1}(a), (b), and (c) show the uncontrolled states,
the controlled states using $\bm{K}_i^*(\bm{\alpha}^*)$ in~(\ref{eq_ex1_K_ori}), and
the controlled states using $\bm{K}_i^*(\bm{\alpha}^*)$ in~(\ref{eq_ex1_gain}), respectively.
It is clear that
the uncontrolled states are bounded away from  zero.
In contrast, we have
the controlled states
$x_1(t)\approx x_2(t)\approx x_3(t)\approx 0.$
The system states in this example have been effectively
controlled, illustrating the validity of the proposed MO design.
Furthermore,
referring to Figs.~\ref{fig_ex1}(b) and (c) associated with the controller gains in~(\ref{eq_ex1_K_ori}) and~(\ref{eq_ex1_gain}), respectively,
we conclude that
 our MO approach can  reduce ``the values of gain matrices'' by adding an additional term to the objective function
 as in~(\ref{eq_ex1_MOMIP_aug})
 while maintaining excellent performance of the controlled system.

\begin{figure*}
\centering
\begin{equation*}
\begin{array}{ccc}
 \includegraphics[width=5.6cm]{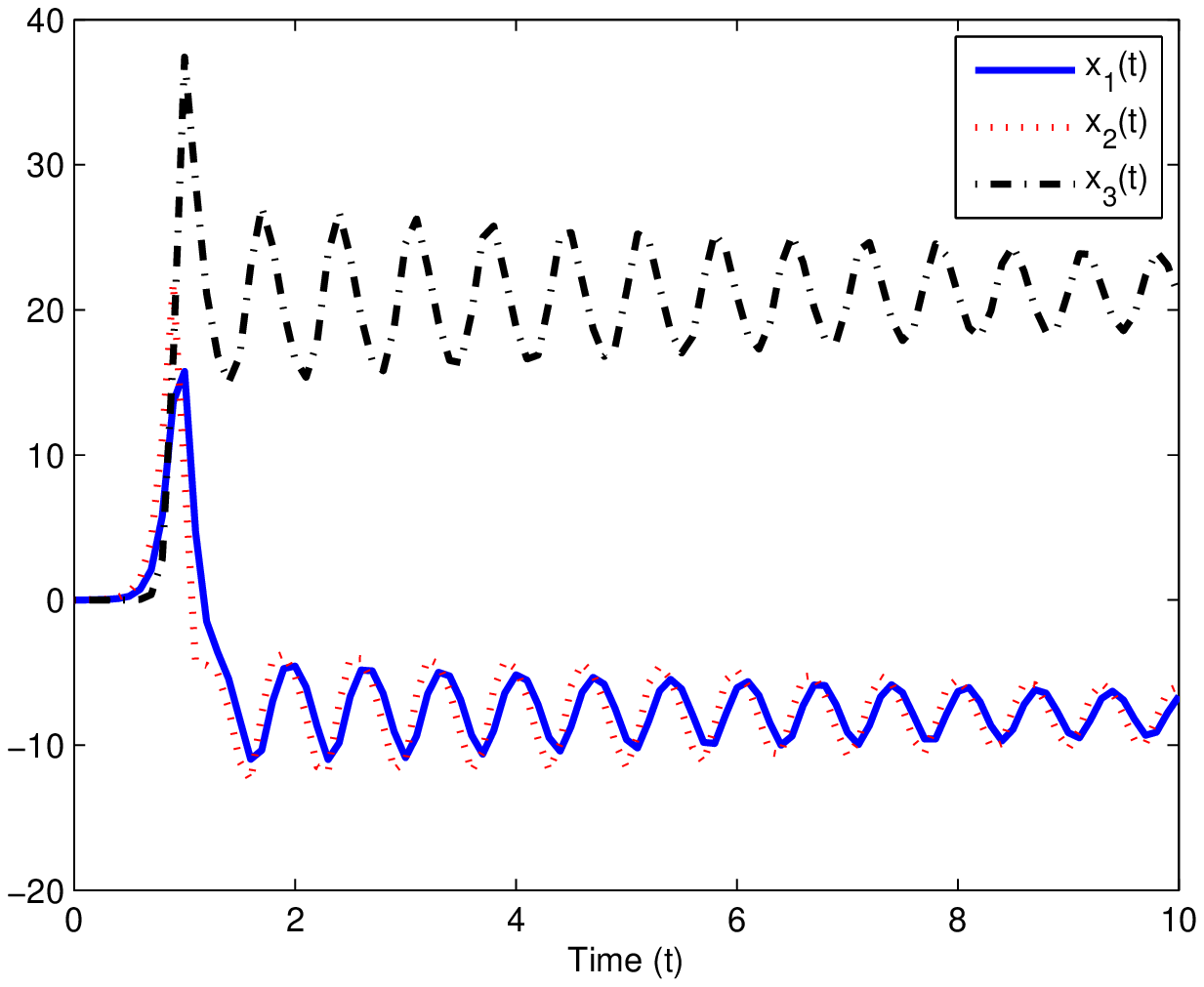}  &\includegraphics[width=5.6cm]{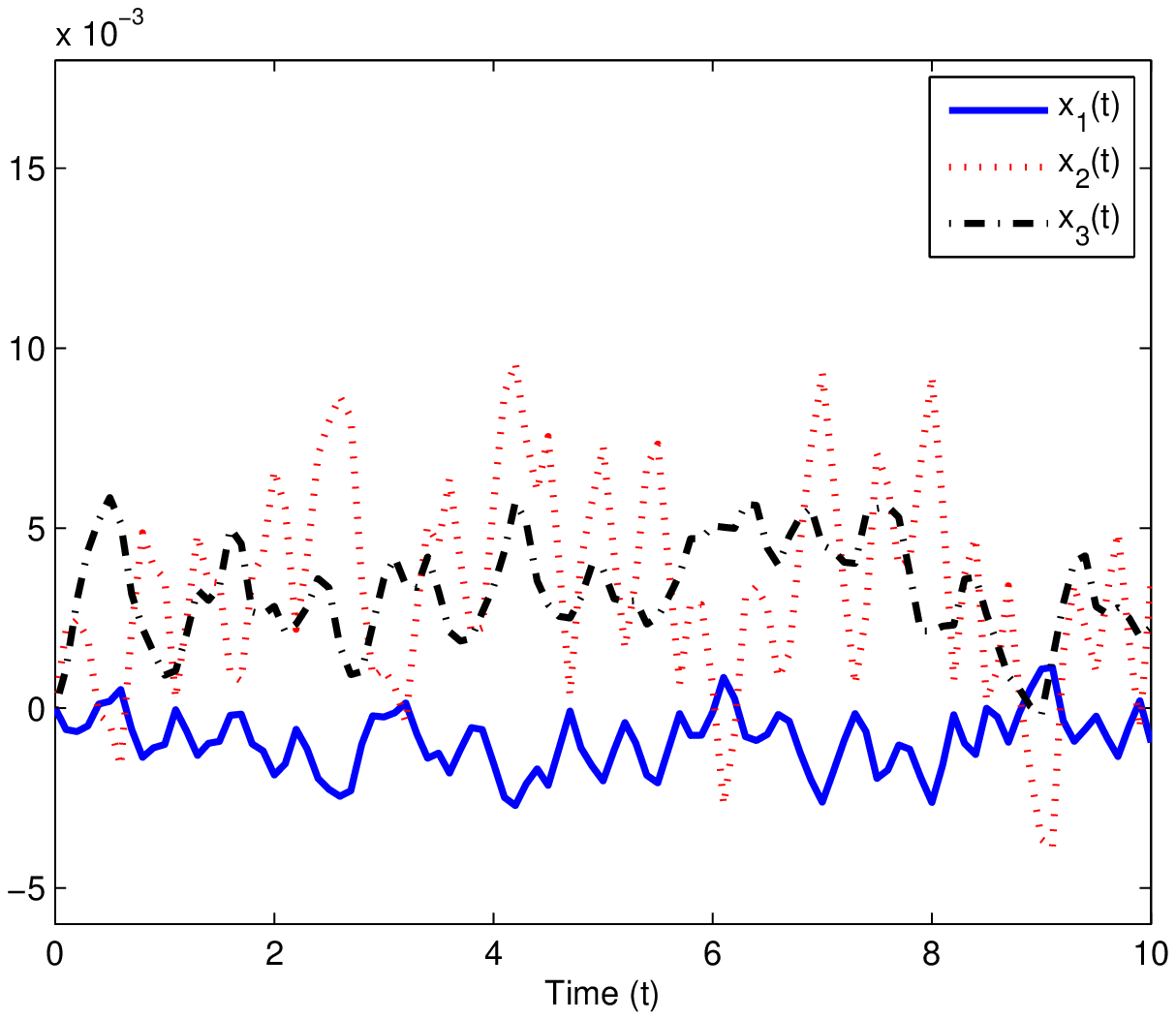} & \includegraphics[width=5.6cm]{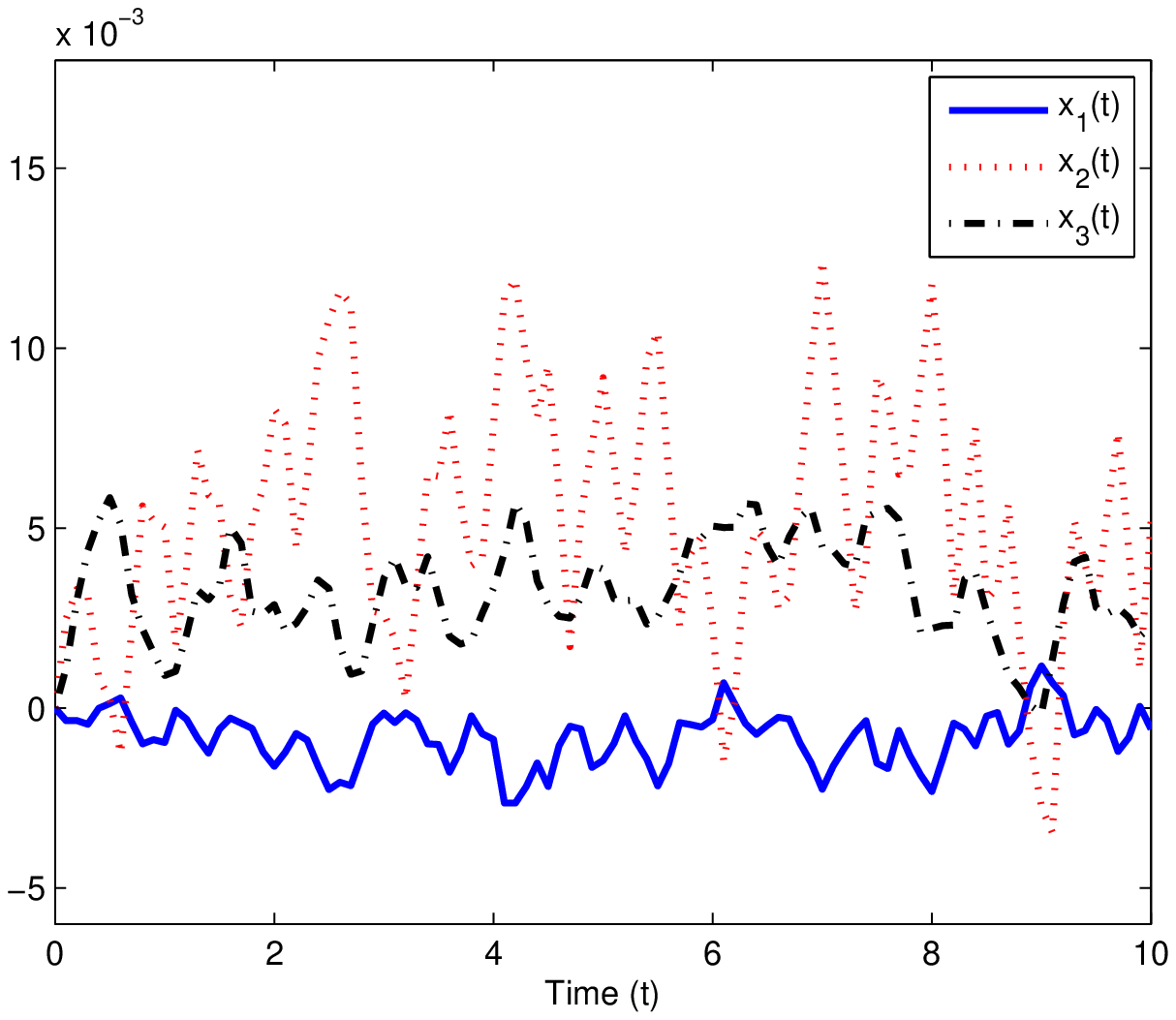} \\
  \mbox{(a)} & \mbox{(b)} & \mbox{(c)}
\end{array}
\end{equation*}
\caption{Comparison between  uncontrolled and controlled  states in Example~1: (a) $u(t)=0$; (b) $u(t)$ constructed according to~(\ref{eq_ut_ex1}) with $\bm{K}_i^*(\bm{\alpha})$ in~(\ref{eq_ex1_K_ori});
and (c) $u(t)$ constructed according to~(\ref{eq_ut_ex1}) with $\bm{K}_i^*(\bm{\alpha})$ in~(\ref{eq_ex1_gain}).
Note that the scale along the y-axis in~(b) and (c) is $10^{-3}$.
While the uncontrolled states in~(a) are bounded away from the zero,
the controlled states in~(b) and~(c) vibrate around the zero with very small amplitudes, indicating that
the proposed MO design approach can yield feasible control laws.
Although the differences between~(b) and~(c)
are insignificant, the associated controller gains in~(c) are smaller than those in~(b).
}\label{fig_ex1}
\end{figure*}

\subsection{Example 2: BIBO System Design}

\begin{figure}
  \centering
  \includegraphics[width=6cm]{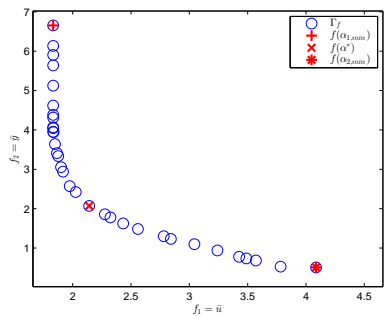}\\
  \caption{APF $\Gamma_f$ in Example 2. Three designs, denoted by  $\bm{\alpha}_{1,min},\bm{\alpha}^*$
and $\bm{\alpha}_{2,min}$, have been selected for comparison. The designs $\bm{\alpha}_{1,min}$
and $\bm{\alpha}_{2,min}$ aim to minimise the upper bounds on  the input norm $||\bm{u}(t)||$ and the output norm $||y(t)||$, respectively.
In contrast, the proposed MO approach yields the design $\bm{\alpha}^*$ lying in the knee region of the APF,
which corresponds to a tradeoff solution that  sacrifices both objectives to some degree in order to do well on the whole.
}\label{fig_ex3_PF}
\end{figure}

For the linear system in~(\ref{eq_m3}),
the following values have been chosen in our simulations:
\begin{equation*}
\bm{x}(0){ }={ }
\left[
  \begin{array}{c}
    3 \\
    -4 \\
  \end{array}
\right],
\bm{A}= \left[
      \begin{array}{cc}
        -10  & -5 \\
       -4   & -1.2 \\
      \end{array}
    \right],
\bm{B}{ }={ }\left[
      \begin{array}{cc}
        3  & 1 \\
       0   & 2 \\
      \end{array}
    \right],\mbox{and }
 \bm{C}=
 \left[
   \begin{array}{cc}
     1 & 0.7 \\
   \end{array}
 \right].
\end{equation*}
 Fig.~\ref{fig_ex3_PF} presents the APF
resulting from solving~(\ref{eq_ex3_MOMIP}).
The APF shows  that
 better performance (small $\bar{y}$) comes from the cost of using more input energy (large $\bar{u}$).
The nondominated vector
\begin{equation*}
\bm{f}(\bm{\alpha}^*)=\bm{\alpha}^*=
\left[
  \begin{array}{cc}
     2.1412  &  2.0705 \\
  \end{array}
\right]^T
\end{equation*}
has been selected, leading to the controller gain
\begin{equation*}
\bm{K}^*(\bm{\alpha}^*)=
\left[
  \begin{array}{cc}
   -0.2037 &   0.0363 \\
    0.0168 &  -0.4106 \\
  \end{array}
\right].
\end{equation*}
We compared the selected design to two ``extreme'' designs denoted by
\begin{equation*}
\bm{\alpha}_{1,min}{ }={ }
\left[
  \begin{array}{cc}
     1.8324 &   6.6537 \\
  \end{array}
\right]^T \mbox{ and }
\bm{\alpha}_{2,min}{ }={ }
\left[
  \begin{array}{cc}
   4.0862  &  0.5084 \\
  \end{array}
\right]^T
\end{equation*}
 as shown in Fig.~\ref{fig_ex3_PF}.
 The corresponding gain matrices are
 \begin{equation*}
\bm{K}^*(\bm{\alpha}_{1,min})=
\left[
  \begin{array}{cc}
   -0.0813 &   0.1873 \\
    0.0979 &  -0.2918 \\
  \end{array}
\right] \mbox{ and }
\bm{K}^*(\bm{\alpha}_{2,min})=
\left[
  \begin{array}{cc}
   -3.0274 &  -2.3897 \\
    0.0858 &  -0.7900 \\
  \end{array}
\right].
\end{equation*}

Table~\ref{tab_cpr} summarizes the performance of the three designs, i.e., $\bm{\alpha}^*,\bm{\alpha}_{1,min}$,
and $\bm{\alpha}_{2,min}$.
As expected, the design $\bm{\alpha}_{1,min}$ leads to the smallest $\max_t || \bm{u}(t) ||$
because it aims at minimising the upper bound $\bar{u}$ on  $||\bm{u}(t)||$.
The design $\bm{\alpha}_{2,min}$ leads to the smallest $\max_t |y(t)|$
as this design focuses on the minimisation of the upper bound $\bar{y}$ on  $||y(t)||$.
In contrast, the proposed approach provides a trade-off design that
sacrifices both objectives to some degree in order to do well on the whole.

\begin{table}
  \centering
  \caption{Performance Comparison of the Three Tradeoff Designs}\label{tab_cpr}
  \begin{tabular}{cccc}
    \hline
    \hline
    Performance Index$\backslash$ Design  & $\bm{\alpha}_{1,min}$ &  $\bm{\alpha}^*$ & $\bm{\alpha}_{2,min}$ \\
        \hline
   $\max_t ||\bm{u}(t)||$ &    1.7674  &   1.8906  &   3.6689 \\
       \hline
   $\max_t |y(t)|$ &     1.0476  &   0.8751  &   0.2000 \\
    \hline
    \hline
  \end{tabular}
\end{table}

\section{Conclusion}\label{sec_con}

In this paper,  linear matrix inequality (LMI) approaches have been extended to address controller design problems with multiple objectives.
A hybrid multiobjective differential evolution (HMODE) algorithm, a type of multiobjective evolutionary algorithms (MOEAs), has been proposed.
Although LMI approaches and the development of MOEAs are mature,
we have integrated them for multiobjective (MO) controller designs, which has not been thoroughly investigated in the literature.
Several benefits of using
the integration are summarized as follows.
First, Pareto optimality for MO controller designs can be achieved in a systematic way.
This is done by
solving an MO matrix inequality problem (MOMIP) using the HMODE algorithm.
Second, a system designer can have a clear view on how objectives affect each other.
This comes from the availability of
an approximated Pareto front (APF) produced by solving the MOMIP.
 In this case, a broad perspective on the optimality can be obtained as compared to conventional SO designs.
Third, if the shape of the APF is bent, then it is possible to select the ``optimal''  trade-off controller: all objectives are sacrificed to some degree
but the system performance is improved on the whole, i.e., a win-win situation can be created.
Fourth, design parameters can be selected
through optimisation  without using a heuristic assignment.
Finally, our MO design approach can adjust the values of controller gains by adding an additional term to the cost function, showing its flexibility.

\end{document}